%% file: paper.tex
\documentclass[letterpaper, 10 pt, conference]{ieeeconf}

\IEEEoverridecommandlockouts
\usepackage{cite}

\usepackage{amsthm,amsmath,amssymb,amsfonts}
\usepackage{bm}
\usepackage{algorithmic}
\usepackage[ruled, noline, linesnumbered]{algorithm2e}
\usepackage{graphicx}
\usepackage{xcolor}

\usepackage{array}
\usepackage{multirow}

\usepackage[normalem]{ulem}

\makeatletter
\let\MYcaption\@makecaption
\makeatother

\usepackage[font=footnotesize]{subcaption}

\makeatletter
\let\@makecaption\MYcaption
\makeatother

\usepackage{textcomp}
\def\BibTeX{{\rm B\kern-.05em{\sc i\kern-.025em b}\kern-.08em
T\kern-.1667em\lower.7ex\hbox{E}\kern-.125emX}}

\usepackage{lipsum}


\theoremstyle{plain}
\newtheorem{theorem}{Theorem}
\newtheorem{lemma}{Lemma}
\newtheorem{corollary}{Corollary}

\theoremstyle{definition}
\newtheorem{example}{Example}

\def\({\left(}
\def\){\right)}
\def\[{\left[}
\def\]{\right]}

\def\ubf{{\bf u}}
\def\xbf{{\bf x}}
\def\Phibf{{\bf \Phi}}
\def\Acal{\mathcal{A}}  
\def\Hcal{\mathcal{H}}  
\def\Rcal{\mathcal{R}}  
\def\Scal{\mathcal{S}}  

\def\w#1{\widetilde{#1}}
\def\v#1{\overrightarrow{#1}}
\def\vinv#1{\overleftarrow{#1}}
\def\nul#1{\text{null}\left(#1\right)}

\DeclareMathOperator*{\argmin}{argmin}

\def\alg#1{Algorithm~\ref{alg:#1}}
\def\cor#1{Corollary~\ref{cor:#1}}
\def\fig#1{Fig.~\ref{fig:#1}}
\def\lem#1{Lemma~\ref{lem:#1}}
\def\sec#1{Section~\ref{sec:#1}}
\def\tab#1{Table~\ref{tab:#1}}
\def\thm#1{Theorem~\ref{thm:#1}}
\def\eqn#1{\eqref{eqn:#1}}

\def\mat#1{\begin{bmatrix}#1\end{bmatrix}}
\def\st{{\rm s.t.}}
\def\OptConsSep{&&\quad}

\newcommand{\norm}[1]{\left\lVert#1\right\rVert}

\makeatletter
\newcommand{\OptMin}{\@ifstar\OptMinStar\OptMinNoStar}
\makeatother
\newcommand\OptCons[3]{
&\ #1
\ifx\\#2\\ \else \OptConsSep #2 \fi%
\ifx\\#3\\ \nonumber \else \label{eqn:#3} \fi%
}
\newcommand{\OptMinStar}[3][]{%
\ifx\\#1\\ \else\begin{subequations}\label{eqn:#1}\fi%
\begin{alignat}{2}%
\min\ &\ #2 \nonumber \\%
\st\ #3 %
\end{alignat}%
\ifx\\#1\\ \else\end{subequations}\fi%
}
\newcommand{\OptMinNoStar}[3][]{%
\ifx\\#1\\ \else\begin{subequations}\label{eqn:#1}\fi%
\begin{alignat}{2}%
\min\ &\ #2 \ifx\\#1\\ \nonumber \else \tag{\ref{eqn:#1}}\label{eqnset:#1} \fi \\%
\st\ #3 %
\end{alignat}%
\ifx\\#1\\ \else\end{subequations}\fi%
}

\newcommand{\centeritem}[3]{
\hspace*{#1}#2\hspace*{#3}

}

\newcommand{\centerlegend}[3]{\centeritem{#1}{\includegraphics{#2}}{#3}}

\title{\LARGE \bf System Level Synthesis via Dynamic Programming}

\author{Shih-Hao Tseng, Carmen Amo Alonso, and SooJean Han
\thanks{Shih-Hao Tseng, Carmen Amo Alonso, and SooJean Han are with the Division of Engineering and Applied Science, California Institute of Technology, Pasadena, CA 91125, USA.  Emails: {\tt\small \{shtseng,camoalon,soojean\}@caltech.edu}}
}

\begin{document}

\maketitle
\thispagestyle{empty}
\pagestyle{empty}

\bstctlcite{IEEE_BSTcontrol}

\input{abstract}
\input{introduction}

\input{DP}

\input{H2}

\input{evaluation}

\input{conclusion}

\bibliographystyle{IEEEtran}
\bibliography{Test}

\end{document}

%% file: abstract.tex
\begin{abstract}

System Level Synthesis (SLS) parametrization facilitates controller synthesis for large, complex, and distributed systems by incorporating system level constraints (SLCs) into a convex SLS problem and mapping its solution to stable controller design. Solving the SLS problem at scale efficiently is challenging, and current attempts take advantage of special system or controller structures to speed up the computation in parallel. However, those methods do not generalize as they rely on the specific system/controller properties.

We argue that it is possible to solve general SLS problems more efficiently by exploiting the structure of SLS constraints. In particular, we derive dynamic programming (DP) algorithms to solve SLS problems. In addition to the plain SLS without any SLCs, we extend DP to tackle infinite horizon SLS approximation and entrywise linear constraints, which form a superclass of the locality constraints. Comparing to convex program solver and naive analytical derivation, DP solves SLS $4$ to $12\times$ faster and scales with little computation overhead. We also quantize the cost of synthesizing a controller that stabilizes the system in a finite horizon through simulations.

\end{abstract}

%% file: introduction.tex
\section{Introduction}\label{sec:introduction}
System Level Synthesis (SLS) facilitates the incorporation of system level constraints (SLCs) by its parametrization of closed-loop system responses
\cite{anderson2019system,wang2019system}. Subsequently, SLS greatly simplifies the controller synthesis problem of large-scale, complex, distributed systems into a convex program. Given matrices $A$ and $B$ depending on the system dynamic, SLS formulates the following convex program
\OptMin[SLS]{
g(\Phibf_\xbf,\Phibf_\ubf)
}{
\OptCons{
\mat{zI-A & -B}
\mat{{\Phibf_\xbf}\\
{\Phibf_\ubf}
}
=
I,
}{}{sf-constraints}\\
\OptCons{
{\Phibf_\xbf},{\Phibf_\ubf} \in z^{-1}\Rcal\Hcal_{\infty},
}{}{}\\
\OptCons{
\mat{{\Phibf_\xbf}\\
{\Phibf_\ubf}} \in \Scal,
}{}{system-level-constraints}
}
where $g$ is the objective, $\Scal$ the set of SLCs, and $z^{-1}\Rcal\Hcal_{\infty}$ the set of strictly proper stable transfer functions. The system responses $\{ \Phibf_\xbf, \Phibf_\ubf \}$ are transfer functions given by
\begin{align*}
\Phibf_\xbf = \sum_{\tau = 1}^{T} z^{-\tau} \Phi_x[\tau], \quad\quad \Phibf_\ubf = \sum_{\tau = 1}^{T} z^{-\tau} \Phi_u[\tau]
\end{align*}
with horizon $T$ and spectral components $\Phi_x[\tau]$ and $\Phi_u[\tau]$.
In what follows, we assume that the objective $g$ can be decomposed into a sum of per-step costs
\begin{align*}
g(\Phibf_\xbf,\Phibf_\ubf) = \sum\limits_{\tau = 1}^{T} g_{\tau}(\Phi_x[\tau],\Phi_u[\tau]).
\end{align*}
Once \eqn{SLS} is solved, SLS gives a state-feedback controller $\Phibf_\xbf\Phibf_\ubf^{-1}$ that can stabilize the systems in horizon $T$.

Although SLS transforms a complex controller synthesis problem into a tractable convex program, in practice, solving \eqn{SLS} is still computationally demanding, especially when dealing with large-scale systems. To accelerate the solving process, existing proposals, such as \cite{anderson2017structured,matni2017scalable,wang2018separable,alonso2019distributed}, focus on specially structured systems and controllers (e.g., localizable systems) which allow the decomposition of \eqn{SLS} for parallel processing. However, for general systems without the desired structures, those methods are no longer applicable, and we resort to the heuristics or programming techniques in convex program solvers to speed up the solving process.

We then ask the question: \emph{Is it possible to expedite the solving process without special structural assumptions?} Our answer is affirmative. So far the existing proposals essentially impose various structural constraints through \eqn{system-level-constraints}, and we have not yet fully exploited the structure of the SLS constraint \eqn{sf-constraints}. By treating the system responses as state and control of a system, one can solve the SLS problem by dynamic programming (DP) \cite{bertsekas2005dynamic1}, which is highly efficient. This opportunity is also noticed in \cite{anderson2017structured}, where the authors applied the DP principle for SLS with a quadratic cost. But to do so, \cite{anderson2017structured} requires additional input coupling assumption: The boundary states must be directly controlled by the corresponding boundary actuators. A fully general DP algorithm remains open.

\subsection{Contributions and Organization}
In \sec{DP}, we derive the DP algorithms to solve the SLS problem \eqn{SLS}. Starting with plain SLS without SLCs \eqn{system-level-constraints}, we then provide the approximation to infinite horizon SLS and incorporate a class of SLCs, the entrywise linear constraints, into the DP process. The entrywise linear constraints are generalized versions of the sparsity/locality constraints in the literature \cite{anderson2019system}. To demonstrate the usage of the DP algorithms, we adopt them for the SLS problems with a $\Hcal_2$ objective in \sec{H2}. Through extensive simulations in \sec{evaluation}, we show that DP algorithms are more scalable and can outperform existing convex program solver by $4$ to $12\times$ and naive Lagrange multiplier method by $10$ to $38\times$. We also quantify the synthesis overhead for stabilizing a system in a finite horizon by comparing DP with DP approximation to infinite horizon SLS. Finally, we conclude in \sec{conclusion} with future research directions.

\subsection{Notation}
We use lower- and upper-case letters (such as $x$ and $A$) to denote vectors and matrices respectively, while bold lower- and upper-case characters and symbols (such as $\ubf$ and ${\Phibf_\ubf}$) are reserved for signals and transfer matrices. Let $A^{ij}$ be the entry of $A$ at the $i^{\rm th}$ row and $j^{\rm th}$ column. We denote by $A^+$ the pseudo inverse (Moore-Penrose inverse).
We vectorize a matrix $A$ to be the vector $\v{A}$ by stacking its columns. Inversely, we rebuild the matrix $\vinv{x}$ from a vector $x$ by realigning the elements. The null space of a matrix $\Psi$ is written as $\nul{\Psi} = \{v : \Psi v = 0\}$, where $0$ is an all-zero vector. We slightly abuse the notation to write $A \in \nul{\Psi}$ if all columns in $A$ are in $\nul{\Psi}$.
Let $\left\lVert {\Phibf}_\ubf\right\rVert_{\Hcal_2}^2$ be the $\Hcal_2$ norm of a transfer function $\Phibf_\ubf$, which is given by $\sum_{t=0}^{\infty}\left\lVert{\Phi_u[t]}\right\rVert_{F}^2$ with $\lVert \cdot \rVert_{F}$ the Frobenius norm.

%% file: DP.tex
\section{Dynamic Programming Algorithms}\label{sec:DP}
We illustrate the dynamic programming (DP) algorithms for state-feedback SLS problems.
We begin with DP for plain SLS and reduce it as an approximation to infinite horizon SLS problems. We then extend the DP algorithm to handle entrywise linear constraints.

Before our DP derivations, we introduce two lemmas here to use later. The proofs are trivial and omitted.
\begin{lemma}\label{lem:null-space-is-subspace}
Given a matrix $\Psi$, $\nul{\Psi}$ is a subspace and there exists some matrix $\Xi$ such that
$
\nul{\Psi} = \{ v : v = \Xi \theta \}
$
where $\theta$ is an arbitrary vector.
\end{lemma}
\begin{lemma}\label{lem:union-of-null-space}
The intersection of $\nul{\Psi_a}$ and $\nul{\Psi_b}$ is $\nul{\mat{\Psi_a\\ \Psi_b}}$.
\end{lemma}

\subsection{Plain SLS}\label{sec:DP-plain-SLS}
We first derive the DP algorithm for plain SLS without constraint \eqn{system-level-constraints}. Notice that we can rewrite the SLS constraints \eqn{sf-constraints} in the following form with spectral components:\\[-\baselineskip]
\begin{subequations}\label{eqn:sf-space-time-constraints}
\begin{alignat}{2}
\OptCons{
\Phi_x[\tau+1] = A \Phi_x[\tau] + B \Phi_u[\tau],
}{}{}\\
\OptCons{\hspace*{20ex} \forall \tau = 1, \dots, T-1,}{}{sf-state-dynamic}\\
\OptCons{
\Phi_x[1] = I,
}{}{sf-initial-condition}\\
\OptCons{
A \Phi_x[T] + B \Phi_u[T] = 0.
}{}{sf-boundary-condition}
\end{alignat}
\end{subequations}

Treating $\Phi_x[\tau]$ as state variable $X[\tau]$ and $\Phi_u[\tau]$ as control $U[\tau]$ at each time $\tau$, the SLS problem is equivalent to the following discrete time control problem with state dynamics \eqn{sf-state-dynamic}, initial condition \eqn{sf-initial-condition}, and boundary condition \eqn{sf-boundary-condition}:
\OptMin[equivalent-SLS]{
\sum\limits_{\tau = 1}^{T} g_{\tau}(X[\tau],U[\tau])
}{
\OptCons{
X[\tau + 1] = A X[\tau] + B U[\tau]
,
}{}{equivalent-state-dynamic}\\
\OptCons{\hspace*{20ex} \forall \tau = 1, \dots, T-1,}{}{}\\
\OptCons{
X[1] = I,
}{}{equivalent-initial-condition}\\
\OptCons{
A X[T] + B U[T] = 0.
}{}{equivalent-boundary-condition}
}

To solve the above problem by DP, we have to address the following issues:
\begin{itemize}
\item Compute the cost-to-go function $V_\tau(X[\tau])$ recursively backwards in time according to the state dynamic~\eqn{equivalent-state-dynamic}.
\item Ensure the boundary condition~\eqn{equivalent-boundary-condition} can be satisfied at each step of the backward recursion.
\end{itemize}

The form of the cost-to-go function $V_\tau$ depends on the objective $g$; we will explicitly derive $V_\tau$ for the specific $\Hcal_2$ objective in \sec{H2}. For now, we derive $V_\tau(X[\tau])$ implicitly by definition via the following recursive relationship:
\begin{alignat}{2}
V_{\tau}(X[\tau])
=&\ \min\limits_{\hat{U} \in \Acal_U[\tau]} g_{\tau}(X[\tau],\hat{U}) +&& V_{\tau + 1}(AX[\tau]+B\hat{U}), \nonumber \\
&&&\ \forall \tau = 1,\dots,T
\label{eqn:cost-to-go}
\end{alignat}
where $\Acal_U[\tau]$ denotes the admissible set of $U$ at time $\tau$ and $V_{T+1}(X[\tau]) = 0$.

Accordingly, the control $U[\tau]$ at each time $\tau$ is given by
\begin{align}
U[\tau] =&\ \argmin\limits_{\hat{U} \in \Acal_U[\tau]} g_{\tau}(X[\tau],\hat{U}) + V_{\tau + 1}(AX[\tau]+B\hat{U}) \nonumber \\
=&\ K_\tau (X[\tau]),
\label{eqn:K-function}
\end{align}
and the state $X[\tau]$ follows~\eqn{equivalent-state-dynamic} and~\eqn{equivalent-initial-condition}.

We assign $X[T+1]$ to be an all-zero matrix, so that $X[T+1] \in \nul{I} = \nul{\Psi_x[T+1]}$ and \eqn{equivalent-boundary-condition} is met in the form of \eqn{equivalent-state-dynamic}. Given $X[\tau+1] \in \nul{\Psi_x[\tau + 1]}$, the following theorem then constructs $\Acal_U[\tau]$ and asserts that $X[\tau] \in \nul{\Psi_x[\tau]}$ to satisfy the boundary condition~\eqn{equivalent-boundary-condition} at the end of backward recursion.
\begin{theorem}\label{thm:feasible_set}
Suppose $X[\tau+1] \in \nul{\Psi_x[\tau+1]}$, we have
\begin{align*}
\Acal_U[\tau] = \{ \hat{U} : \hat{U} = Q_X X[\tau] + Q_\Lambda \Lambda \}
\end{align*}
where
\begin{align*}
\Gamma[\tau] =&\ \mat{-B & \Xi_x[\tau + 1]}, \\
Q_X =&\ \mat{I & 0} \Gamma^+[\tau] A, \\
Q_\Lambda =&\ \mat{I & 0} (I - \Gamma^+[\tau] \Gamma[\tau]),
\end{align*}
and $\Xi_x[\tau+1]$ is given by \lem{null-space-is-subspace} for $\nul{\Psi_x[\tau+1]}$.

Also, $X[\tau] \in \nul{\Psi_x[\tau]}$ where
\begin{align*}
\Psi_x[\tau] = (\Gamma[\tau] \Gamma^+[\tau] - I)A.
\end{align*}
\end{theorem}

\begin{proof}
Since $X[\tau+1] \in \nul{\Psi_x[\tau+1]}$, by \lem{null-space-is-subspace}, there exists a matrix $\Xi_x[\tau]$ such that we can express all columns in $X[\tau+1]$ as $\Xi_x[\tau] \theta$ for some vector $\theta$. In other words, there exists a matrix $\Theta$ such that
\begin{align*}
\Xi_x[\tau + 1]\Theta = X[\tau + 1] = AX[\tau] + BU[\tau].
\end{align*}
Rearranging the terms and applying the definition of $\Gamma[\tau]$ yield
\begin{align*}
\mat{-B & \Xi_x[\tau + 1]}\mat{U[\tau] \\ X[\tau + 1]} = \Gamma[\tau] \mat{U[\tau] \\ X[\tau + 1]} = AX[\tau].
\end{align*}
Therefore, we must have
\begin{align}
\mat{U[\tau]\\\Theta} = \Gamma^+[\tau] A X[\tau] + (I - \Gamma^+[\tau] \Gamma[\tau]) \Lambda,
\label{eqn:feasible-U-and-chi}
\end{align}
for some matrix $\Lambda$, or
\begin{align*}
U[\tau] \in \{ \hat{U} : \hat{U} = Q_X X[\tau] + Q_\Lambda \Lambda \} = \Acal_U[\tau].
\end{align*}
To ensure that \eqn{feasible-U-and-chi} is solvable, \cite[Theorem 1]{james1978generalised} asserts that
\begin{align*}
\Gamma[\tau] \Gamma^+[\tau] AX[\tau] = AX[\tau].
\end{align*}
As such, we have
\begin{align*}
X[\tau] \in \nul{(\Gamma[\tau] \Gamma^+[\tau] - I)A} = \nul{\Psi_x[\tau]}.
\end{align*}
\\[-2\baselineskip]
\end{proof}

Together, we summarize the derivation in \alg{DP-plain}.

\begin{algorithm}[!t]
\begin{algorithmic}[1]
\REQUIRE $A, B$ and objective $g(\Phibf_\xbf,\Phibf_\ubf)$.
\ENSURE $\Phi_x[\tau], \Phi_u[\tau]$ for all $\tau = 1, \dots, T$.
\STATE{$\nul{\Psi_x[T+1]} = \nul{I}$.}\label{line:alg-DP-plain-initial-null-Psi}
\STATE{$V_{T+1}(X[\tau]) = 0$.}
\FOR{$\tau = T, \dots, 1$}
\STATE{Derive $\Acal_U[\tau]$ and $\nul{\Psi_x[\tau]}$ from~\thm{feasible_set}.}\label{line:alg-DP-plain-Fcal-U}
\STATE{Compute $K_\tau(X[\tau])$ by \eqn{K-function}.}
\STATE{Derive $V_{\tau}(X[\tau])$ from \eqn{cost-to-go}.}
\ENDFOR
\STATE{$\Phi_x[1] = I$.}
\FOR{$\tau = 1, \dots, T$}
\STATE{$\Phi_u[\tau] = K_\tau(\Phi_x[\tau])$.}
\STATE{$\Phi_x[\tau + 1] = A \Phi_x[\tau] + B\Phi_u[\tau]$.}
\ENDFOR
\end{algorithmic}
\caption{DP for plain SLS.}
\label{alg:DP-plain}
\end{algorithm}

\subsection{Approximation to Infinite Horizon SLS}\label{sec:DP-infinite-horizon-SLS}
For infinite horizon SLS, the constraints~\eqn{sf-constraints} become:
\begin{alignat}{2}
\OptCons{
\Phi_x[\tau+1] = A \Phi_x[\tau] + B \Phi_u[\tau],
}{\forall \tau = 1, \dots, \infty,}{}\\
\OptCons{
\Phi_x[1] = I
}{}{}
\end{alignat}
where the key difference from its finite counterpart~\eqn{sf-space-time-constraints} arises due to the absence of the boundary condition \eqn{sf-boundary-condition}.

A naive approximation to the infinite horizon SLS is to solve \eqn{equivalent-SLS} without the condition \eqn{equivalent-boundary-condition}, which leads to
\OptMin{
\sum\limits_{\tau = 1}^{T} g_{\tau}(X[\tau],U[\tau])
}{
\OptCons{
X[\tau + 1] = A X[\tau] + B U[\tau],
}{\forall \tau = 1, \dots, T-1,}{}\\
\OptCons{
X[1] = I.
}{}{}
}
The above optimization problem can also be solved by dynamic programming. We can obtain the corresponding DP algorithm by relaxing the feasible set $\Acal_U[\tau]$ to be the whole space and removing line \ref{line:alg-DP-plain-initial-null-Psi} and \ref{line:alg-DP-plain-Fcal-U} from \alg{DP-plain}.

\subsection{SLS with Entrywise Linear Constraints}\label{sec:DP-SLC}
We now elaborate on how to incorporate the system level constraints \eqn{system-level-constraints} into DP. Especially, we consider the \emph{entrywise linear constraints} as follows
\begin{align*}
\Scal = \{ \Phibf_\xbf, \Phibf_\ubf : &\ \Phi_x[\tau] \in \Scal_x[\tau],\nonumber\\
&\ \Phi_u[\tau] \in \Scal_u[\tau], \text{ for all } \tau = 1,\dots,T \},
\end{align*}
where
\begin{subequations}\label{eqn:entrywise-linear-constraint}
\begin{align}
\Scal_x[\tau] =&\ \{ \Phi : \v{\Phi} = S_x[\tau] \v{\Phi} \},
\label{eqn:constraints-Sx}\\
\Scal_u[\tau] =&\ \{ \Phi : \v{\Phi} = S_u[\tau] \v{\Phi} \}.
\label{eqn:constraints-Su}
\end{align}
\end{subequations}

The entrywise linear constraint \eqn{entrywise-linear-constraint} allows the entries in each spectral component to depend linearly on one another. It generalizes the sparsity constraints in the literature \cite{anderson2019system}, which confines the non-zero entries of each spectral component.

\begin{example}[Sparsity as Entrywise Linear Constraints]
Let $\Phi_x[\tau]$ be a $2\times 2$ matrix. Consider the sparsity constraint which dictates $\Phi_x^{12}[\tau] = 0$. We can express the sparsity constraint as a entrywise linear constraint by enforcing a binary diagonal $S_x[\tau]$ with $0$ at the corresponding entries:
\begin{align*}
\v{\Phi_x[\tau]} =&\ S_x[\tau] \v{\Phi_x[\tau]} \quad \Leftrightarrow\\
\mat{
\Phi_x^{11}[\tau] \\
\Phi_x^{12}[\tau] \\
\Phi_x^{21}[\tau] \\
\Phi_x^{22}[\tau]
}
=&\ \mat{
1 & 0 & 0 & 0 \\
0 & 0 & 0 & 0 \\
0 & 0 & 1 & 0 \\
0 & 0 & 0 & 1
}
\mat{
\Phi_x^{11}[\tau] \\
\Phi_x^{12}[\tau] \\
\Phi_x^{21}[\tau] \\
\Phi_x^{22}[\tau]
}
=
\mat{
\Phi_x^{11}[\tau] \\
0 \\
\Phi_x^{21}[\tau] \\
\Phi_x^{22}[\tau]
}.
\end{align*}
\end{example}

Entrywise linear constraints are more general than sparsity or locality constraints in that entries can exist beyond the diagonal or binary values.

To incorporate entrywise linear constraints
into the DP algorithm, we consider the following vectorized version of~\eqn{equivalent-SLS}:
\OptMin{
\sum\limits_{\tau = 1}^{T} g_{\tau}(x[\tau],u[\tau])
}{
\OptCons{
x[\tau + 1] = \w{A} x[\tau] + \w{B} u[\tau],
}{\forall \tau = 1, \dots, T-1,}{}\\
\OptCons{
x[1] = \v{I},
}{}{}\\
\OptCons{
\w{A} x[T] + \w{B} u[T] = 0,
}{}{}
}
where $x[\tau] = \v{\Phi_x[\tau]}$ and $u[\tau] = \v{\Phi_u[\tau]}$. $\w{A}$ and $\w{B}$ are defined such that $\w{A}x[\tau] = A \Phi_x[\tau]$ and $\w{B}u[\tau] = B \Phi_u[\tau]$.

We first incorporate the condition \eqn{constraints-Su}. By~\lem{null-space-is-subspace}, we can express:
\begin{align}
\Scal_u[\tau] =&\ \{ \Phi : \v{\Phi} \in \nul{S_u[\tau] - I} \} \nonumber \\
=&\ \{ \Phi : \v{\Phi} = \Xi_{S_u}[\tau] \hat{v} \}.
\label{eqn:null-space-Su}
\end{align}
Therefore, we can enforce~\eqn{constraints-Su} by requiring $u[\tau] = \Xi_{S_u}[\tau] v[\tau]$. Similar to the plain SLS case, we can derive $v[\tau]$, $u[\tau]$, and the cost-to-go function $V_{\tau}(x[\tau])$ by
\begin{align}
&v[\tau] = \argmin\limits_{\hat{v} \in \Acal_v[\tau]} g_{\tau}(x[\tau],\Xi_{S_u}[\tau]\hat{v}) + V_{\tau + 1}(\w{A}x[\tau]+\w{B}\Xi_{S_u}[\tau]\hat{v}) \nonumber\\
&u[\tau] = \Xi_{S_u}[\tau] v[\tau] = K_\tau(x[\tau]),
\label{eqn:vectorized-K-function} \\
&V_{\tau}(x[\tau]) = g_{\tau}(x[\tau],u[\tau]) + V_{\tau + 1}(\w{A}x[\tau]+\w{B}u[\tau]),
\label{eqn:vectorized-cost-to-go}
\end{align}
for all $\tau = T-1, \dots, 1$ and $V_{T+1}(x[\tau]) = 0$.

Again, let $x[T+1] = 0 \in \nul{I}$ to enforce the boundary condition, we derive the admissible set $\Acal_v[\tau]$ and some condition for $x[\tau]$ from a corollary to \thm{feasible_set} below.

\begin{corollary}\label{cor:feasible_set}
Suppose $x[\tau+1] \in \nul{\Psi_x[\tau+1]}$, we have
\begin{align*}
\Acal_v[\tau] = \{ \hat{v} : \hat{v} = Q_x x[\tau] + Q_\lambda \lambda \}
\end{align*}
where
\begin{align*}
\Gamma[\tau] =&\ \mat{-\w{B}\Xi_{S_u}[\tau] & \Xi_x[\tau + 1]}, \\
Q_x =&\ \mat{I & 0} \Gamma^+[\tau] \w{A}, \\
Q_\lambda =&\ \mat{I & 0} (I - \Gamma^+[\tau] \Gamma[\tau]),
\end{align*}
and $\Xi_x[\tau+1]$ is given by \lem{null-space-is-subspace} for $\nul{\Psi_x[\tau+1]}$.

Also, $x[\tau] \in \nul{\Omega[\tau]}$ where
\begin{align*}
\Omega[\tau] = (\Gamma[\tau] \Gamma^+[\tau] - I)\w{A}.
\end{align*}
\end{corollary}

We omit the proof of \cor{feasible_set}, which is the same as the proof of~\thm{feasible_set}.

We also need to enforce the condition \eqn{constraints-Sx}. We know
\begin{align*}
x[\tau] \in \Scal_x[\tau] =&\ \{ \Phi : \v{\Phi} \in \nul{S_x[\tau] - I} \}.
\end{align*}
Meanwhile, \cor{feasible_set} suggests that $x[\tau] \in \nul{\Omega[\tau]}$. So by \lem{union-of-null-space}, we have $x[\tau] \in \nul{\Psi_x[\tau]}$ where
\begin{align}
\nul{\Psi_x[\tau]} = \nul{\mat{S_x[\tau] - I \\ \Omega[\tau]}}.
\label{eqn:vectorized-null-Psi}
\end{align}

Together, we summarize the derivation in \alg{DP-SLC}.

\begin{algorithm}[!t]
\begin{algorithmic}[1]
\REQUIRE $A, B$, objective $g(\Phibf_\xbf,\Phibf_\ubf)$, and entrywise linear constraints $\Scal$.
\ENSURE $\Phi_x[\tau], \Phi_u[\tau]$ for all $\tau = 1, \dots, T$.
\STATE{$\nul{\Psi_x[T+1]} = \nul{I}$.}
\STATE{$V_{T+1}(x[\tau]) = 0$.}
\FOR{$\tau = T, \dots, 1$}
\STATE{Derive $\Xi_{S_u}[\tau]$ from \eqn{null-space-Su}.}
\STATE{Derive $\Acal_v[\tau]$ and $\nul{\Omega[\tau]}$ from \cor{feasible_set}.}
\STATE{Compute $K_\tau(x[\tau])$ by \eqn{vectorized-K-function}.}
\STATE{Derive $V_{\tau}(x[\tau])$ from \eqn{vectorized-cost-to-go}.}
\STATE{Obtain $\nul{\Psi_x[\tau]}$ by \eqn{vectorized-null-Psi}.}
\ENDFOR
\STATE{$\Phi_x[1] = I$.}
\FOR{$\tau = 1, \dots, T$}
\STATE{$\v{\Phi_u[\tau]} = K_\tau(\v{\Phi_x[\tau]})$.}
\STATE{$\Phi_x[\tau + 1] = A \Phi_x[\tau] + B\Phi_u[\tau]$.}
\ENDFOR
\end{algorithmic}
\caption{DP for SLS with entrywise linear constraints.}
\label{alg:DP-SLC}
\end{algorithm}

%% file: H2.tex
\section{Case Study: $\Hcal_2$ Objective}\label{sec:H2}

In \sec{DP}, we derived DP algorithms for plain SLS, an approximation to infinite horizon SLS, and SLS with sparsity constraints. Here, we apply the algorithms to a specific objective function, the $\Hcal_2$ objective, as an example. The $\Hcal_2$ objective function is given by
\begin{align*}
g(\Phibf_\xbf,\Phibf_\ubf) = \norm{C \Phibf_\xbf +D \Phibf_\ubf}_{\Hcal_2}^2 = \sum\limits_{\tau = 1}^{T} g_{\tau}(\Phi_x[\tau],\Phi_u[\tau])
\end{align*}
where
\begin{align*}
g_{\tau}(\Phi_x[\tau],\Phi_u[\tau]) =  \norm{C \Phi_x[\tau] +D \Phi_u[\tau]}_F^2
\end{align*}
for some matrices $C$ and $D$.

\renewcommand{\paragraph}[1]{\vspace{0.2cm}\textbf{#1}}

\paragraph{Plain SLS:}
We apply \alg{DP-plain} to plain SLS with $\Hcal_2$ objective. The first step is to derive an explicit cost-to-go function $V_{\tau}(X[\tau])$, and we proceed by mathematical induction. We claim
\begin{align}
V_\tau(X[\tau]) = \sum\limits_{t =\tau}^{T} \norm{ P_t[\tau] X[\tau] }_F^2
\label{eqn:H2-cost-to-go}
\end{align}
where $P_t[\tau]$ are some matrices.

For $\tau = T+1$, $V_{T+1}(X[\tau]) = 0$ satisfies the claim.

For $\tau < T+1$, \eqn{cost-to-go} gives
\begin{align*}
&\ V_{\tau}(X[\tau]) = \min\limits_{\hat{U} \in \Acal_U[\tau]} \left\lbrace
\norm{ CX[\tau] + D \hat{U} }_F^2 \vphantom{\sum\limits_{t =\tau + 1}^{T}} \right. \\
&\quad\quad + \left. \sum\limits_{t =\tau + 1}^{T} \norm{ P_t[\tau+1] (AX[\tau] + B \hat{U}) }_F^2 \right\rbrace.
\end{align*}
By \thm{feasible_set}, we have
\begin{align}
&\ V_{\tau}(X[\tau]) = \min\limits_{\Lambda} \left\lbrace
\norm{ C_X X[\tau] + D_\Lambda \Lambda }_F^2 \vphantom{\sum\limits_{t =\tau + 1}^{T}} \right. \nonumber \\
&\quad\quad + \left. \sum\limits_{t =\tau + 1}^{T} \norm{ P_t[\tau+1] (A_X X[\tau] + B_\Lambda \Lambda) }_F^2 \right\rbrace
\label{eqn:H2-cost-to-go-before-optimization}
\end{align}
where $A_X = A + B Q_X$, $B_\Lambda = B Q_\Lambda$, $C_X = C + D Q_X$, and $D_\Lambda = D Q_\Lambda$.

We can find the minimizer $\Lambda[\tau]$ to \eqn{H2-cost-to-go-before-optimization} by the first-order condition of its derivative:
\begin{align*}
&\ D_\Lambda^\top (C_X X[\tau] + D_\Lambda \Lambda[\tau]) \\
&\ + B_\Lambda^\top \sum\limits_{t =\tau + 1}^{T} P_t[\tau+1]^\top P_t[\tau+1] (A_X X[\tau] + B_\Lambda \Lambda[\tau]) = O
\end{align*}
where $O$ is the all-zero matrix.

By defining
\begin{align}
P[\tau] =&\ \sum\limits_{t = \tau + 1}^{T} P_t[\tau + 1]^\top P_t[\tau + 1], \label{eqn:H2-P}\\
L[\tau] =&\ - (D_\Lambda^\top D_\Lambda + B_\Lambda^\top P[\tau] B_\Lambda)^{-1} (D_\Lambda^\top C_X + B_\Lambda^\top P[\tau] A_X),\nonumber
\end{align}
we can derive
\begin{align}
\Lambda[\tau] =&\ L[\tau] X[\tau] \nonumber \\
U[\tau] =&\ Q_X X[\tau] + Q_\Lambda \Lambda[\tau] = (Q_X + Q_\Lambda L[\tau]) X[\tau] \nonumber \\
=&\ K_\tau(X[\tau]) = K[\tau]X[\tau]
\label{eqn:H2-K-function}
\end{align}
where we introduce the matrix $K[\tau]$ accordingly.

Consequently, the cost-to-go function is
\begin{align}
&\ V_{\tau}(X[\tau]) =
\norm{ (C + D K[\tau]) X[\tau] }_F^2 \nonumber \\
&\quad\quad + \sum\limits_{t =\tau + 1}^{T} \norm{ P_t[\tau+1] (A + B K[\tau])X[\tau] }_F^2,
\label{eqn:H2-derived-cost-to-go}
\end{align}
of which the form matches our claim \eqn{H2-cost-to-go}.

Although we need the cost-to-go function for derivation, we don't need it when deriving $U[\tau]$. As such, we only need to keep track of the quantity $P[\tau]$. Compare \eqn{H2-cost-to-go} and \eqn{H2-derived-cost-to-go}, we can update \eqn{H2-P} by
\begin{align}
P[\tau - 1] =&\ (C + D K[\tau])^\top (C + D K[\tau]) \nonumber \\
&\ + (A + B K[\tau])^\top P[\tau] (A + B K[\tau]).
\label{eqn:H2-update-P}
\end{align}

In sum, we derive \alg{DP-H2} for plain SLS with $\Hcal_2$ objective.

\begin{algorithm}[!t]
\begin{algorithmic}[1]
\REQUIRE $A, B, C, D$.
\ENSURE $\Phi_x[\tau], \Phi_u[\tau]$ for all $\tau = 1, \dots, T$.
\STATE{$\nul{\Psi_x[T+1]} = \nul{I}$.}
\STATE{$P = 0$.}
\FOR{$\tau = T, \dots, 1$}
\STATE{Derive $\Acal_U[\tau]$ and $\nul{\Psi_x[\tau]}$ from \thm{feasible_set}.}
\STATE{Compute $K[\tau]$ by \eqn{H2-K-function}.}
\STATE{Update $P$ by \eqn{H2-update-P}.}
\ENDFOR
\STATE{$\Phi_x[1] = I$.}
\FOR{$\tau = 1, \dots, T$}
\STATE{$\Phi_u[\tau] = K[\tau]\Phi_x[\tau]$.}
\STATE{$\Phi_x[\tau + 1] = A \Phi_x[\tau] + B\Phi_u[\tau]$.}
\ENDFOR
\end{algorithmic}
\caption{DP for plain SLS with $\Hcal_2$ objective.}
\label{alg:DP-H2}
\end{algorithm}

\paragraph{Approximation to Infinite Horizon SLS:}
We then extend the DP approximation to infinite horizon SLS to $\Hcal_2$ objective. As discussed in \sec{DP-infinite-horizon-SLS}, we just need to relax the feasible set $\Acal_U[\tau]$ to be the whole space, and the cost-to-go function becomes
\begin{align*}
V_{\tau}(X[\tau]) =&\ \min\limits_{\hat{U}} \left\lbrace
\norm{ CX[\tau] + D \hat{U} }_F^2 \vphantom{\sum\limits_{t =\tau + 1}^{T}} \right. \\
&\ \left. + \sum\limits_{t =\tau + 1}^{T} \norm{ P_t[\tau+1] (AX[\tau] + B \hat{U}) }_F^2 \right\rbrace.
\end{align*}
Similarly, we obtain the first-order condition of the derivative:
\begin{align*}
D^\top (C X[\tau] + D U[\tau]) + B^\top P[\tau] (A X[\tau] + B U[\tau]) = O
\end{align*}
where $P[\tau]$ is defined in \eqn{H2-P}. Hence,
\begin{align}
K[\tau] =&\ -(D^\top D + B^\top P[\tau] B)^{-1} (D^\top C + B^\top P[\tau] A), \label{eqn:H2-Approx-K-function}\\
U[\tau] =&\ K[\tau] X[\tau] = K_{\tau}(X[\tau]) \nonumber
\end{align}
and the corresponding DP Approx algorithm is summarized in \alg{DP-Approx-H2}.

\begin{algorithm}[!t]
\begin{algorithmic}[1]
\REQUIRE $A, B, C, D$.
\ENSURE $\Phi_x[\tau], \Phi_u[\tau]$ for all $\tau = 1, \dots, T$.
\STATE{$P = 0$.}
\FOR{$\tau = T, \dots, 1$}
\STATE{Compute $K[\tau]$ by \eqn{H2-Approx-K-function}.}
\STATE{Update $P$ by \eqn{H2-update-P}.}
\ENDFOR
\STATE{$\Phi_x[1] = I$.}
\FOR{$\tau = 1, \dots, T$}
\STATE{$\Phi_u[\tau] = K[\tau]\Phi_x[\tau]$.}
\STATE{$\Phi_x[\tau + 1] = A \Phi_x[\tau] + B\Phi_u[\tau]$.}
\ENDFOR
\end{algorithmic}
\caption{DP Approx: DP approximation for infinite horizon SLS with $\Hcal_2$ objective.}
\label{alg:DP-Approx-H2}
\end{algorithm}

\paragraph{SLS with Entrywise Linear Constriants:}
Finally, we specialize \alg{DP-SLC} for $\Hcal_2$ objective.
We first introduce $\w{C}$ and $\w{D}$ such that
\begin{align*}
C \Phi_x[\tau] = \w{C} \v{\Phi_x[\tau]}, \quad\text{and}\quad
D \Phi_u[\tau] = \w{D} \v{\Phi_u[\tau]},
\end{align*}
for all $\tau$. Likewise, we assume that
\begin{align*}
V_\tau(x[\tau]) = \sum\limits_{t =\tau}^{T} \norm{ P_t[\tau] x[\tau] }_F^2.
\end{align*}

Following the similar procedure, we obtain $\Xi_{S_u}[\tau]$ from \eqn{null-space-Su} and compute
\begin{align*}
L_n[\tau] =&\ D_\lambda[\tau]^\top C_x[\tau] + B_\lambda[\tau]^\top P[\tau] A_x[\tau],\\
L_d[\tau] =&\ D_\lambda[\tau]^\top D_\lambda[\tau] + B_\lambda[\tau]^\top P[\tau] B_\lambda[\tau],\\
L[\tau] =&\ - L_d[\tau]^{-1} L_n[\tau],
\end{align*}
where $P[\tau]$ is defined in \eqn{H2-P} and
\begin{align*}
A_x[\tau] =&\ \w{A} + \w{B} \Xi_{S_u}[\tau] Q_x, & B_\lambda[\tau] =&\ \w{B} \Xi_{S_u}[\tau] Q_\lambda, \\
C_x[\tau] =&\ \w{C} + \w{D} \Xi_{S_u}[\tau] Q_x, & D_\lambda[\tau] =&\ \w{D} \Xi_{S_u}[\tau] Q_\lambda.
\end{align*}
Accordingly,
\begin{align}
\lambda[\tau] =&\ L[\tau] x[\tau], \nonumber \\
v[\tau] =&\ Q_x[\tau] x[\tau] + Q_\lambda[\tau] \lambda[\tau] = (Q_x[\tau] + Q_\lambda[\tau] L[\tau]) x[\tau], \nonumber\\
u[\tau] =&\ \Xi_{S_u}[\tau] v[\tau] = \Xi_{S_u}[\tau] (Q_x[\tau] + Q_\lambda[\tau] L[\tau]) x[\tau] \nonumber\\
=&\ K_\tau(x[\tau]) = K[\tau]x[\tau]
\label{eqn:H2-SLC-K-function}
\end{align}
with $K[\tau]$ defined correspondingly.

As a result,
\begin{align*}
&\ V_{\tau}(x[\tau]) =
\norm{ (\w{C} + \w{D} K[\tau]) x[\tau] }_F^2 \nonumber \\
&\quad\quad + \sum\limits_{t =\tau + 1}^{T} \norm{ P_t[\tau+1] (\w{A} + \w{B} K[\tau])x[\tau] }_F^2,
\end{align*}
which confirms our assumption above. Therefore, we can update $P[\tau]$ by
\begin{align}
P[\tau - 1] =&\ (\w{C} + \w{D} K[\tau])^\top (\w{C} + \w{D} K[\tau]) \nonumber \\
&\ + (\w{A} + \w{B} K[\tau])^\top P[\tau] (\w{A} + \w{B} K[\tau])
\label{eqn:H2-SLC-update-P}
\end{align}
and we summarize the derivation in \alg{DP-SLC-H2}.

\begin{algorithm}[!t]
\begin{algorithmic}[1]
\REQUIRE $A, B, C, D$ and $S_x[\tau], S_u[\tau]$
\ENSURE $\Phi_x[\tau], \Phi_u[\tau]$ for all $\tau = 1, \dots, T$.
\STATE{$\nul{\Psi_x[T+1]} = \nul{I}$.}
\STATE{$P = 0$.}
\FOR{$\tau = T, \dots, 1$}
\STATE{Derive $\Xi_{S_u}[\tau]$ from \eqn{null-space-Su}.}
\STATE{Derive $\Acal_v[\tau]$ and $\nul{\Omega[\tau]}$ from \cor{feasible_set}.}
\STATE{Compute $K[\tau]$ by \eqn{H2-SLC-K-function}.}
\STATE{Update $P$ by \eqn{H2-SLC-update-P}.}
\STATE{Obtain $\nul{\Psi_x[\tau]}$ by \eqn{vectorized-null-Psi}.}
\ENDFOR
\STATE{$\Phi_x[1] = I$.}
\FOR{$\tau = 1, \dots, T$}
\STATE{$\v{\Phi_u[\tau]} = K[\tau]\v{\Phi_x[\tau]}$.}
\STATE{$\Phi_x[\tau + 1] = A \Phi_x[\tau] + B\Phi_u[\tau]$.}
\ENDFOR
\end{algorithmic}
\caption{DP for SLS with $\Hcal_2$ objective subject to entrywise linear constraints.}
\label{alg:DP-SLC-H2}
\end{algorithm}

%% file: evaluation.tex
\section{Evaluation}\label{sec:evaluation}

We evaluate our algorithms through simulations. We first compare the scalability of DP against existing solver CVX \cite{CVX} and naive Lagrange multiplier method
, which also yields the analytical solution. We then simulate DP (\alg{DP-H2}) and DP Approx (\alg{DP-Approx-H2}) to evaluate the cost, in terms of computation overhead, of obtaining a finite impulse response (FIR) system under feedback. We begin with our simulation setup and a brief introduction of the Lagrange multiplier method.

\subsection{Simulation Setup and Naive Lagrange Multiplier Method}
We conduct the simulations using SLSpy \cite{tsengsubslspy,SLSpy}. In each simulation, we synthesize controllers for $100$ random systems and collect the statistical data. Each system is a fully actuated chain with $N_x$ nodes, where the $A$ matrix is tridiagonal with randomly generated off-diagonal entries, and $B$ matrix is a diagonal matrix with random diagonal entries. The SLS objective is as follows
\begin{align}
g(\Phibf_\xbf, \Phibf_\ubf)
= \left\lVert \mat{I \\ 0} \Phibf_\xbf + \mat{0 \\ I} \Phibf_\ubf \right\rVert_{\Hcal_2}^2
= \left\lVert \mat{\Phibf_\xbf \\ \Phibf_\ubf} \right\rVert_{\Hcal_2}^2,
\label{eqn:simulation-objective}
\end{align}
where $C$ and $D$ matrices are defined accordingly. We consider the d-locality constraint as in \cite{wang2016localized} with actuation delay $1$, communication delay $2$ and $d = 3$. As a subclass of the sparsity constraints, we can also express locality constraints as entrywise linear constraints. The results are measured on a desktop with AMD Ryzen 7 3700X processor ($16$ logical cores) and $32$ GB DDR4 memory.

Since the SLS constraints \eqn{sf-space-time-constraints} and the locality constraints (in the form of \eqn{entrywise-linear-constraint}) are all equalities, we can rewrite the SLS problem as
\begin{align*}
\min\ g(\Phibf_\xbf,\Phibf_\ubf) \ \ \st\ h(\Phibf_\xbf,\Phibf_\ubf) = 0
\end{align*}
and apply the naive Lagrange multiplier method to solve
\begin{align*}
\nabla_{\Phibf_\xbf, \Phibf_\ubf, \lambda}\ g(\Phibf_\xbf,\Phibf_\ubf) - \lambda h(\Phibf_\xbf,\Phibf_\ubf) = 0.
\end{align*}
Given the objective \eqn{simulation-objective}, we express the above condition as
\begin{align*}
J \v{\Phi} - b = 0
\end{align*}
for some matrix $J$ and vector $b$, where $\v{\Phi}$ is a vector of the entries in $\Phibf_\xbf$ and $\Phibf_\ubf$, and compute $\v{\Phi}$ by $J^{-1} b$.

\subsection{Scalability with System Size}
To evaluate the scalability of the methods, we run the simulations with different system size $N_x$, measure the average synthesis time for the plain SLS and SLS with locality constraints, and summarize the results in \tab{scalability-H2} and \tab{scalability-H2-localized}, respectively. Among the methods, DP scales the best. For plain SLS, DP is $12$ to $22 \times$ faster than CVX and $10$ to $4000\times$ faster than naive Lagrange multiplier method; With locality constraints, DP is $4$ to $17\times$ faster than CVX and more than $38\times$ faster than naive Lagrange multiplier method, which cannot even deal with $N_x = 20$. We remark that DP outperformed two other methods using only one CPU core without any optimization, while CVX parallelized its work over $16$ logical cores. It is possible to improve the performance of DP by parallelizing its computation.

\begin{table}
\centering
\caption{Average synthesis time of plain SLS for random chain-like systems.}
\label{tab:scalability-H2}
\renewcommand{\arraystretch}{1.25}
\begin{tabular}{|
@{}>{\centering}m{0.08\columnwidth}@{}||
@{}>{\centering}m{0.3\columnwidth}@{}|
@{}>{\centering}m{0.3\columnwidth}@{}|
@{}>{\centering}m{0.3\columnwidth}@{}|}
\hline
\multirow{3}{*}{$N_x$} & \multicolumn{3}{c|}{Synthesis Time (ms)} \tabularnewline
\cline{2-4}
& DP & \multirow{2}{*}{CVX} & Lagrange \tabularnewline
& (\alg{DP-H2}) & & Multiplier \tabularnewline
\hline
\hline
\input{exp1-Nx-overhead-table}
\end{tabular}
\smallskip
\caption{Average synthesis time of SLS with locality constraints for random chain-like systems.}
\label{tab:scalability-H2-localized}
\renewcommand{\arraystretch}{1.25}
\begin{tabular}{|
@{}>{\centering}m{0.08\columnwidth}@{}||
@{}>{\centering}m{0.3\columnwidth}@{}|
@{}>{\centering}m{0.3\columnwidth}@{}|
@{}>{\centering}m{0.3\columnwidth}@{}|}
\hline
\multirow{3}{*}{$N_x$} & \multicolumn{3}{c|}{Synthesis Time (ms)} \tabularnewline
\cline{2-4}
& DP & \multirow{2}{*}{CVX} & Lagrange \tabularnewline
& (\alg{DP-SLC-H2}) & & Multiplier \tabularnewline
\hline
\hline
\input{exp1-localized-Nx-overhead-table}
\end{tabular}
\end{table}

\subsection{Cost for FIR System}
The boundary constraint is essential for the synthesized controller to stabilize a system in a finite horizon (FIR system). When the desired horizon goes to infinity, the controllers subject to the boundary condition become the ones without. Below, we examine the computation overhead for an FIR system and evaluate how close the DP controller (by \alg{DP-H2}) is to the DP Approx controller (by \alg{DP-Approx-H2}), which is an approximation to infinite horizon SLS.

\begin{figure}
\centering
\includegraphics[scale=1]{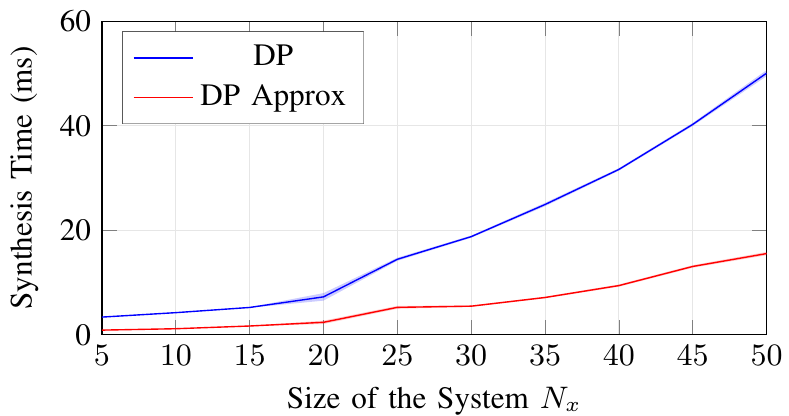}
\includegraphics[scale=1]{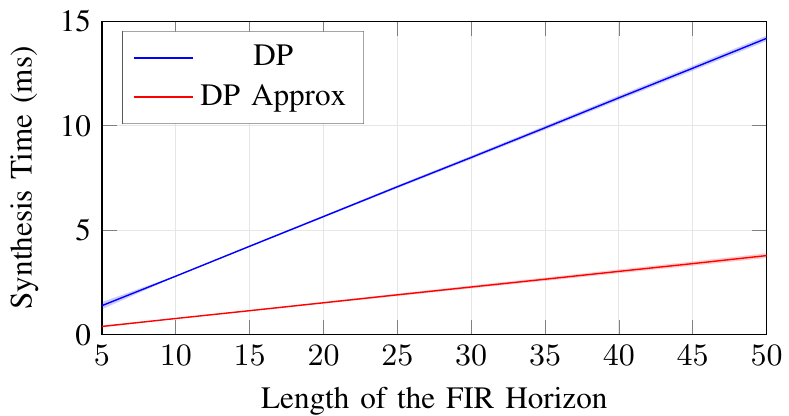}
\caption{Synthesis time of the controller by DP (\alg{DP-H2}) and the infinite horizon approximation controller by DP Approx (\alg{DP-Approx-H2}).}
\label{fig:exp2-overhead}
\end{figure}

\fig{exp2-overhead} shows the computation overhead in terms of synthesis time versus the system size ($N_x$) and the horizon of the synthesized controllers. DP Approx is about $3$ to $4\times$ faster than DP, and both of them scales linearly with the FIR horizon as expected. In exchange, \fig{exp2-horizon} shows that the DP Approx controller fails to stabilize the system within the desired FIR horizon after an impulse noise hits the center of the system. When we allow a longer FIR horizon, DP Approx controller performs the same as the DP controller. In sum, we pay some tens of milliseconds more to guarantee a controller stabilizing a system in a finite horizon.

\begin{figure}
\def\figlmargin{25.01pt}
\def\figrmargin{5.59pt}
\centering
\centerlegend{\figlmargin}{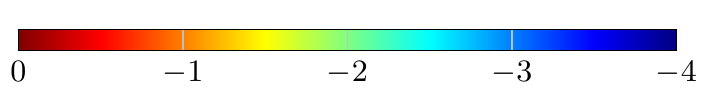}{\figrmargin}
\vspace*{0.25\baselineskip}
\centeritem{\figlmargin}{\small FIR horizon $= 5$.}{\figrmargin}
\includegraphics{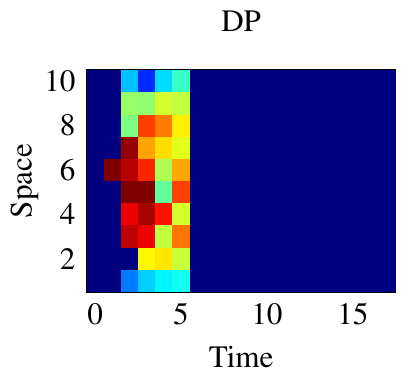}
\hfill
\includegraphics{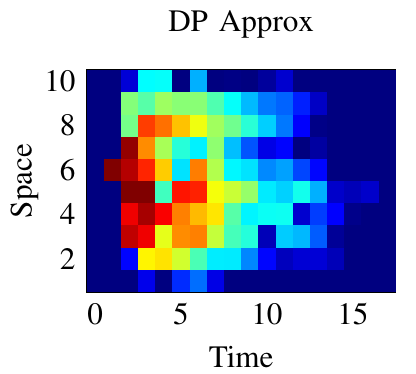}
\\[0.25\baselineskip]
\centeritem{\figlmargin}{\small FIR horizon $= 10$.}{\figrmargin}
\includegraphics{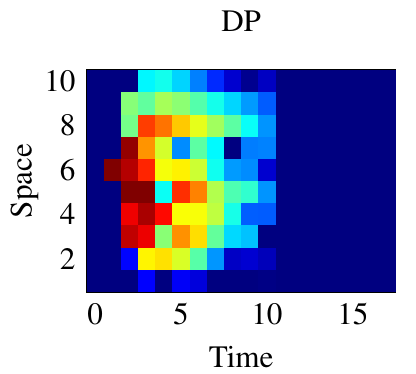}
\hfill
\includegraphics{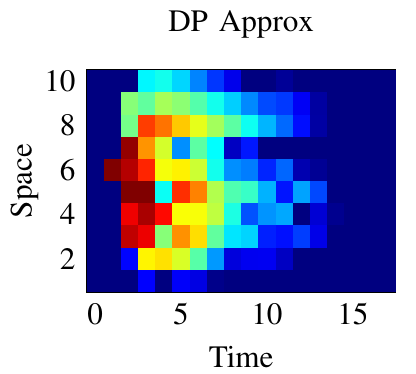}
\\[0.25\baselineskip]
\centeritem{\figlmargin}{\small FIR horizon $= 15$.}{\figrmargin}
\includegraphics{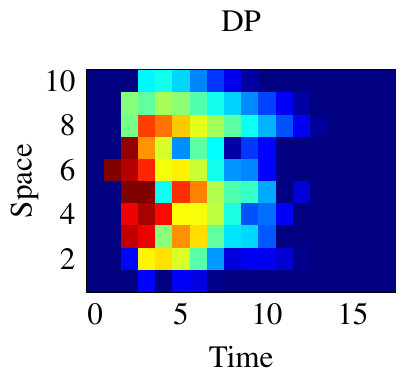}
\hfill
\includegraphics{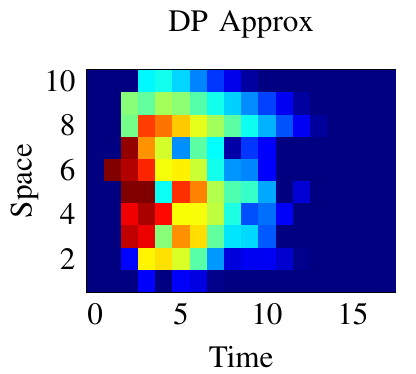}
\caption{Comparison of controllers computed with and without boundary constraints for different time horizons. On the left, the controllers are computed using DP, i.e., accounting for boundary constraints. On the right, the controllers are computed using DP Approx, i.e., no boundary constraints are imposed. Each row corresponds to a different FIR horizon. The figures show the space-time diagram of the state (in log scale, $\log_{10}(|x|)$) after some noise hits the center of the chain at time $0$. The DP Approx controller approximates finite horizon DP controller when the horizon gets longer.}
\label{fig:exp2-horizon}
\end{figure}

%% file: exp1-Nx-overhead-table.tex
$5$ & $5.11$ & $72.44$ & $54.86$\tabularnewline \hline
$10$ & $6.60$ & $90.55$ & $1176.51$\tabularnewline \hline
$15$ & $8.42$ & $136.76$ & $9586.75$\tabularnewline \hline
$20$ & $11.25$ & $244.37$ & $45782.15$\tabularnewline \hline

%% file: exp1-localized-Nx-overhead-table.tex
$5$ & $12.35$ & $217.81$ & $479.60$\tabularnewline \hline
$10$ & $129.48$ & $1411.00$ & $78405.76$\tabularnewline \hline
$15$ & $685.03$ & $4890.28$ & $1013700.11$\tabularnewline \hline
$20$ & $1968.85$ & $8549.75$ & not feasible \tabularnewline \hline

%% file: conclusion.tex
\section{Conclusion and Future Directions}\label{sec:conclusion}

We derived DP algorithms to solve general state-feedback SLS problems, including plain SLS, infinite horizon approximation, and sparsity constrained SLS. Sparsity constraints generalize locality constraints by allowing linear dependency among entries of spectral components. Using $\Hcal_2$ objective as an example, we demonstrate how to adapt DP algorithms to a given objective. Our simulation results suggest that DP significantly outperforms CVX and naive Lagrange multiplier method. We also quantify the computation overhead of obtaining a controller for an FIR system after feedback.

Future work includes extensions of the DP algorithms for output-feedback SLS, which contains more parameters to handle. Also, it is worth covering more constraint classes, such as inequality constraints or dependencies among entries from different spectral components. Finally, one can apply DP to online synthesis settings such as model predictive control, where the computational overhead is crucial.